\documentclass{sig-alternate}
\begin{document}
\newcommand*{\code}{\tt}
\newcommand*{\ser}[2]{#1_1,\dots,#1_#2} 
\newcommand*{\implying}{\!\Rightarrow\!}
\newcommand*{\Eff}   [2]{\mathsf E_\mathrm{#1} [#2]}
\newcommand*{\Statef}[1]{\mathsf \Sigma_\mathrm{#1}}
\newcommand*{\State} [2]{\mathsf \Sigma_\mathrm{#1} [#2]}
\newcommand*{\wAIwA}{({w\mbox{\small A}},\varvar I_{w\mbox{\scriptsize A}})}
\newcommand*{\opname}[1]{\textsf{#1}}
\newcommand*{\arname}[1]{\ensuremath{\mathrm{#1}}}
\newcommand*{\gtag}[1]{\mbox{\sffamily\bfseries #1}}
\newcommand*{\garg}[1]{$\mathrm{#1}$}
\newcommand*{\varvar}[1]{\ensuremath{\mbox{\bfseries \itshape #1}}}
\newcommand*{\nt}[1]{\mbox{\it{#1}}}
\newcommand*{\cond}[3]{(#1\to#2:#3)}
\newcommand*{\real}{\;\mbox{double}\;}
\newcommand*{\id}[1]{\mathrm{#1}}
\newcommand*{\node} {\gtag{node}}
\newcommand*{\dom}  {\gtag{dom}}
\newcommand*{\ports}{\gtag{ports}}
\newcommand*{\body} {\gtag{body}}
\newcommand*{\eval} {\gtag{eval}}
\newcommand*{\Dom} {\opname{Dom}}
\newcommand*{\FP} {{\mathrm F}_P}
\newcommand*{\GP} {{\mathrm G}_P}
\newcommand*{\term}[2]{#1\{#2\}}
\newcommand*{\pred}[3]{#1^{#2}\{#3\}}
\newcommand*{\figname}[1]{#1.pdf}

%
%
%
%
\title{Yet Another Way of Building Exact Polyhedral Model for Weakly Dynamic Affine Programs}
%
%
\numberofauthors{1}
\author{
\alignauthor
Arkady V. Klimov\\
       \affaddr{Institute of Design Problems in Microelectronics}\\
       \affaddr{Russian Academy of Sciences}\\
       \affaddr{Moscow, Russia}\\
       \email{arkady.klimov@gmail.com}
}
%
%
%

\maketitle              

\begin{abstract}

Exact polyhedral model (PM) can be built in the general case if the only control structures are {\code do}-loops and structured {\code if}s, and if loop counter bounds, array subscripts and {\code if}-conditions are affine expressions of enclosing loop counters and possibly some integer constants. In more general dynamic control programs, where arbitrary {\code if}s and {\code while}s are allowed, in the general case the usual dataflow analysis can be only fuzzy. This is not a problem when PM is used just for guiding the parallelizing transformations, but is insufficient for transforming source programs to other computation models (CM) relying on the PM, such as our version of dataflow CM or the well-known KPN.

The paper presents a novel way of building the exact polyhedral model and an extension of the concept of the exact PM, which allowed us to add in a natural way all the processing related to the data dependent conditions. Currently, in our system, only arbirary {\code if}s (not {\code while}s) are allowed in input programs. The resulting polyhedral model can be easily put out as an equivalent program with the dataflow computation semantics.

\keywords{exact polyhedral model, weakly dynamic affine programs, data dependent conditionals, dataflow computation model }

\end{abstract}
\section{Introduction}
\label{s:intro}
The exact Polyhedral Model (PM), a.k.a. Exact Array Dataflow Analysis (EADA), can be build for a limited class of programs. Normally, it embraces affine loop nests with assignments in between, in which loop bounds and array element indexes are affine expressions of surrounding loop variables and fixed structure parameters (array sizes etc.). {\code If}-statements with affine conditions are also allowed. Methods of EADA are well developed \cite{Feautrier88parametricinteger,Feautrier91dataflowanalysis,DBLP:conf/sc/Pught91} for this class of programs. The results are usually used for guiding parallelizing transformations.

We define (see Section \ref{s:form}) the PM as a mapping that assigns to each read (load) instance in the computation its unique write (store) instance that has written the value being read. In other words, it is a collection of source functions, each bound to a single read operation in a program. Such function takes iteration vector of the read instance and produces the name and the iteration vector of a write instance or symbol $\bot$ indicating that such write do not exist and the original state of memory is read. 

However when the source program contains also one or several {\code if}-statements with non-affine (e.g., data dependent) conditions the known methods suggest only approximation which, generally, provides {\em a set of possible} writes for some reads. It is usually referred to as Fuzzy Array Dataflow Analysis (FADA) \cite{Feautrier97FADA}. In some specific cases such model may provide a source function that uses as its input also values of predicates associated with non-affine conditionals in order to produce the unique source. These cases seem to be those in which the number of such predicate values is finite (uniformly bounded).

Usually, this does not makes a problem as parallelization can proceed relying on the approximate PM. But our aim is to convert the source program (part) completely into the dataflow computation model, and any approximation is unacceptable for us. So, our task was to extend the class of programs for which exact PM can be built by programs with non-affine conditionals. Our model representation language is extended with predicate symbols corresponding to non-affine Boolean expressions in the source code. From such a model the exact source function for each read can be easily extracted.  These functions are, generally, recursive and they depend on usual affine parameters as well as on an unlimited number of predicate values. 

But the source function is not our aim. For building dataflow program we need the inversed, \emph{use set} function.  From the parallelization perspective this program carries implicitly the maximum amount of parallelism of the source program. A simple computation strategy (see Section \ref{ss:prop:dflow}) exhibits all this parallelism. More details and references can be found in \cite{ArK:PSI:11, ArKlimov:meta:14}.

In this paper we describe briefly our original way of building the dataflow model for affine programs and then expand it to programs with non-affine conditionals. The affine class and the affine solution tree are defined in Section \ref{s:form}. Sections \ref{s:effect}--\ref{s:bdf} describe our algorithm of building the PM. Several examples are presented in Section \ref{s:example}. Section \ref{s:prop} describes two different semantics of the PM considered as a program. Section \ref{s:relwrk} compares our approach and results with related ones.

\section{Some Formalism}\label{s:form}
Consider a Fortran program fragment $P$. We are interested in memory reads and writes which have the form of array element or scalar variable. The latter will be treated below as 0-dimension arrays. 

We define a {\em computation graph} by running the program $P$ with some input data. The graph consists of two kinds of nodes: reads and writes, corresponding respectively to individual executions of load or store memory operation. There is a link from a write $w$ to a read $r$ if $r$ reads the value written by $w$. Thus, $r$ uses the same memory cell as $w$ and $w$ is the last write to this cell before $r$.

Our purpose is to obtain a compact parametric description of all computation graphs for a given program $P$. To make such description feasible we need to consider a limited class of programs. It is a well known affine class, which can be formally defined by the set of rules shown in Fig.\ref{fig:affprog}.

\begin{figure}[htb]
\begin{center}
\begin{tabular}{@{\hspace*{0mm}}l@{\hspace*{10mm}}l@{\hspace*{0mm}}}
  ${\rm\Lambda}$                              &(empty statement)      \\
  $ {\arname A}(\ser ik)=e$                   &(assignment, $k\ge 0$) \\
  $X_1$; $X_2$                                &(sequence)             \\
  \code if $c$ then $X_1$; else $X_2$; endif  &(conditional)          \\
  \code do \varvar v $=e_1,e_2$; $X$; enddo	  &({\code do}-loop)      \\
\end{tabular}
\end{center}
\caption {Affine program constructors}\label{fig:affprog}
\end{figure}

The right hand side $e$ of an assignment may contain array element access ${\arname A}(\ser ik)$, $k\ge 0$. All index expressions as well as bounds $e_1$ and $e_2$ of \mbox{\code do}-loops must be affine in surrounding loop variables and structure parameters. {\em Affine} expressions are those built from variables and integer constants with addition, subtraction and multiplication by literal integer. Also, in an affine expression, we allow whole division by literal integer. Condition $c$ also must be affine, i.e.,  equivalent to $e=0$ or $e>0$ where $e$ is affine. 

Programs (or program parts) following these limitations have been called {\em static control programs} (SCoP) \cite{Feautrier97FADA,GrieblHabThesis}. Their computation graph depends only on structure parameters and does not depend on dynamic data values. 

Below, we are to remove the restriction that conditional expression $c$ must be affine. Such extended program class has been called {\em weakly dynamic programs} (WDP) \cite{StefanovThesis}. Here we shall allow only arbitrary {\code if}s but not {\code while}s which will be considered in the future.

A point in the computation trace of an affine program may be identified as $(s,I_s)$, where $s$ is a point in the program and $I_s$ is the iteration vector, i.e., a vector of integer values of all enclosing loop variables of point $s$. The list of these variables will be denoted as $\varvar I_s$, which allows to depict the point $s$ itself as $(s,\varvar I_s)$. (Here and below boldface symbols denote  variables or list of variables as syntactic objects, while normal italic symbols denote some values as usual).

Thus, denoting an arbitrary read or write instance as $(r,I_r)$ or $(w,I_w)$ respectively, we represent the whole computation graph as a mapping:
\begin{equation} 
  \FP : (r,I_r) \mapsto (w,I_w) 
\end{equation}
which, for any read node $(r,I_r)$, yields the write node $(w,I_w)$ that has written the value being read, or yields $\bot$ if no such write exist and thus the original contents of the cell is read. This form of graph is called a \emph{source} graph, or S-graph.

However, for translation to dataflow computation model we need the reversed map, that, for each write node, finds all read nodes which read the very value written. So, we need the multi-valued mapping
\begin{equation} 
  \GP : (w,I_w) \mapsto \{(r,I_r)\} 
\end{equation}
which for each write node $(w,I_w)$ yields a set of all read nodes $\{(r,I_r)\}$ that read that very value written. We call this form of computation graph a \emph{use} graph, or U-graph.

A subgraph of S-graph (U-graph) associated with a given read $r$ (write $w$) will be called {\em $r$-component} ({\em $w$-component}).

For each program statement (or point) $s$ we define the domain $\Dom(s)$ as a set of values of iteration vector $I_s$, such that $(s,I_s)$ occurs in the computation. The following proposition summarizes the well-established property of static control programs \cite{Feautrier88parametricinteger,Feautrier91dataflowanalysis,GrieblHabThesis,DBLP:conf/popl/Maslov94,DBLP:conf/sc/Pught91} (which is also justified by our algorithm).
\newtheorem{proposition}{Proposition}
\begin{proposition}
For any statement $(\!s,\!\varvar I_s\!)$ in a static control program $P$ its domain $\Dom(s)$ can be represented as finite disjoint union $\bigcup_i D_i$, such that each subdomain $D_i$ can be specified as a conjunction of affine conditions of variables $\varvar I_s$ and structure parameters, and, when the statement is a read $(r,\varvar I_r)$, there exist such $D_i$ that the mapping $\FP$ on each subdomain $D_i$ can be represented as either $\bot$ or $(w, (\ser em))$ for some write $w$, where each $e_i$ is an affine expression of variables $\varvar I_r$ and structure parameters. \label{prop1}
\end{proposition}

This result suggests the idea to represent each $r$-component of $\FP$ as a solution tree with affine conditions at branching vertices and terms of the form $\term W{e_1, \dots, e_m}$ or $\bot$ at leaves. A similar concept of quasi-affine solution tree, {\em quast}, was suggested by P. Feautrier \cite{Feautrier91dataflowanalysis}. 

A single-valued solution tree (S-tree) is a structure used to represent $r$-components of a S-graph. Its syntax is shown in Fig.\ref{fig:synstree}. It uses just linear expressions (\nt{L-expr}) in conditions and term arguments, so a special vertex type was introduced in order to implement integer division.
\begin{figure}[htb]
\begin{center} \small
\begin{tabular}{@{\hspace*{0mm}}l@{\hspace*{0mm}}l@{\hspace*{1.5mm}}l}
\nt{S-tree}  &  ::=    $ \bot $                                                         &                 \\
\qquad $\mid$&  $\nt{term} $                                                            &                  \\
\qquad $\mid$&  $\cond{\nt{cond}}{\nt{S-tree}_t}{\nt{S-tree}_f} $                       & (branching)       \\
\qquad $\mid$&  $(\nt{L-expr} =: \nt{num}\;\nt{var}+\nt{var}\to \nt{S-tree}) $          & (integer division) \\
\nt{term}    &\ ::=    $ \term{\nt{name}}{\nt{L-expr}_1 ,\dots, \nt{L-expr}_k} $        & ($k\ge 0$)          \\
\nt{var}     &\ ::=    $ \nt{name} $                                                    &                      \\
\nt{num}     &\ ::=    $ \dots \mid -2 \mid -1 \mid 0 \mid 1 \mid 2 \mid 3 \mid \dots $ &                       \\
\nt{cond}    &\ ::=    $ \nt{L-cond} \mid \nt{predicate} $                              & (any condition)        \\
\nt{L-cond}  &\ ::=    $ \nt{L-expr} = 0 \mid \nt{L-expr} > 0 $                         & (affine condition)      \\
\nt{L-expr}  &\ ::=    $ \nt{num} \mid \nt{num}\;\nt{var} +  \nt{L-expr} $              & (affine expression)      \\
\nt{atom}    &\ ::=    $ \bot \mid \term{\nt{name}}{\nt{num}_1 ,\dots, \nt{num}_k} $    & (ground term, $k\!\ge\!0$)\\
\end{tabular}
\end{center}
\caption {Syntax for single-valued solution tree}\label{fig:synstree}
\end{figure}

Given concrete integer values of all free variables of the S-tree one can evaluate the tree to an atom. The two following evaluation rules must be applied iteratively.

A branching like $\cond{c}{T_1}{T_2}$ evaluates to $T_1$ if conditional expression c evaluates to true, otherwise to $T_2$. Non-affine conditions are expressed by a \nt{predicate}. This will be explained below in Section \ref{ss:eff:na}.  

A division $(e =: m \varvar q + \varvar r \to T)$ introduces two new variables $(\varvar q,\varvar r)$ that take respectively the quotient and the remainder of integer division of integer value of $e$ by positive constant integer $m$. The tree evaluates as $T$ with parameter list extended with values of these two new variables.

It follows from Proposition \ref{prop1} that for an affine program $P$ the $r$-component of the S-graph $\FP$ for each read $(r,\varvar I_r)$ can be represented in the form of S-tree $T$ depending on variables $\varvar I_r$ and structure parameters.

However the concept of S-tree is not sufficient for representing $w$-components of U-graph, because those must be multi-valued functions in general. So, we extend the definition of S-tree to the definition of multi-valued tree, M-tree, by two auxiliary rules shown on Fig \ref{fig:synmtree}.

\begin{figure}[htb]
\begin{center} \small
\begin{tabular}{@{\hspace*{0mm}}l@{\hspace*{1mm}}r@{\hspace*{1mm}}l@{\hspace*{4mm}}l}
\nt{M-tree} & ::=    &$ \dots \mbox{the same as for S-tree} \dots $   &                          \\
            & $\mid$ &$ (\& \nt{M-tree}_1 \dots \nt{M-tree}_n) $        & (finite union, $n \ge 2$) \\
            & $\mid$ &$ (\verb"@" \,\nt{var} \to \nt{M-tree}) $         & (infinite union)           \\
\end{tabular}
\end{center}
\caption {Syntax for multi-valued tree}\label{fig:synmtree}
\end{figure}

The semantics also changes. The result of evaluating M-tree is a set of atoms. Symbol $\bot$ now represents the empty set, and the term $\term N{\dots}$ represents a singleton. 

To evaluate $(\& \ser Tn)$ one must evaluate sub-trees $T_i$ and take the union of all results. The result of evaluating $(\verb"@"\varvar v \to T)$ is mathematically defined as the union of infinite number of results of evaluating $T$ with each integer value $v$ of variable $\varvar v$. In practice the result of evaluating $T$ is non-empty only within some bound interval of values $v$. In both cases the united subsets are supposed to be disjoint. 

Below we present our algorithm of building a S-graph (Sections \ref{s:effect} and \ref{s:state}) and then a U-graph (Section \ref{s:bdf}).  

\section{Building Statement Effect}\label{s:effect}
\subsection{Statement Effect and its Evaluation}\label{ss:eff:eval}

Consider a program statement $X$, which is a part of an affine program $P$, and some $k$-dimensional array {\arname A}. Let $\wAIwA$ denote an arbitrary write operation on an element of array {\arname A} within a certain execution of statement $X$, or the totality of all such operations. Suppose that the body of $X$ depends affine-wise on free parameters $\ser pl$ (in particular, they may include variables of loops surrounding $X$ in $P$). We define the effect of $X$ over array {\arname A} as a function 
$$
\Eff A{X} : (\ser pl; \ser qk) \mapsto \wAIwA + \bot
$$
that, for each tuple of parameters $\ser pl$ and indexes $\ser qk$ of an element of array {\arname A}, yields an atom $\wAIwA$ or $\bot$. The atom indicates that the write operation $\wAIwA$ is the last among those that write to element {\arname A}($\ser qk$) during execution of $X$ with affine parameters $\ser pl$ and $\bot$ means that there are no such operations. 

The following claim is another form of Proposition \ref{prop1}: \emph {the effect can be represented as an S-tree with program statement labels as term names}. We call them simply \emph{effect trees}. (All assignments are supposed labeled during preprocessing).

Building effect is the core of our approach. Using S-trees as data objects we implemented some operations on them that are used in the algorithm presented on Fig.\ref{fig:effrules}. A good mathematical foundation of similar operations for similar trees has been presented in \cite{Guda:13}.

The algorithm goes upwards along the AST from primitives like empty and assignment statements. Operation \opname{Seq} computes the effect of a statement sequence from the effects of component statements. Operation \opname{Fold} builds the effect of a {\code do}-loop given the effect of the loop body. For {\code if}-statement with affine condition the effect is built just by putting the effects of branches into a new conditional node.
\begin{figure}[htb]
\begin{center} \small
\begin{tabbing} 
zzzzzzz\=zzzzzzzzzzzzzzzzzzzzzzzzzzzzzzzzzzzzz\=\kill
$\Eff A{\mathrm \Lambda}  =  \bot$				                        \>\> (empty statement) \\
$\Eff A{X_1;X_2}  =  \mbox{\opname{Xeq}} (\Eff A{X_1},\Eff A{X_2})$	    \>\> (sequence)\\
$\Eff A{LA:\; {\arname A}(\ser ek)=e} =  $                              \>\> (assignments to {\arname A})\\
      \> $\cond{\varvar q_1=e_1}{\dots \cond{\varvar q_k=e_k}{\term{LA}{\varvar I}}{\bot}\dots}{\bot}$   \\
      \> where $\varvar I$ is a list of all outer loop variables          \> \\
$\Eff A{LB:\; {\arname B}(\dots)=e}  =  \bot $			                \>\> (other assignments)\\
$\Eff A{\mbox{\code if $c$ then $X_1$ else $X_2$ endif}}= $             \>\> (conditional)\\
      \> $= \cond{c}{\Eff A{X_1}}{\Eff A{X_2}}$                              \\
$\Eff A{\mbox{\code do $\varvar v=e_1,e_2$; $X$; enddo}}  = $           \>\> ({\code do}-loop)\\
      \> $= \mbox{\opname{Fold}}(\varvar v,e1,e2,\Eff A{X}) $      
\end{tabbing}
\end{center}
\caption {The rules for computing effect tree over $k$-dimensional array {\arname A}}\label{fig:effrules}
\end{figure}

The implementation of function \opname{Seq} is straight. To compute \opname{Seq}($T_1$,$T_2$) we simply replace all $\bot$-s in $T_2$ with a copy of $T_1$. Then the result is simplified by a function \opname{Prune} which prunes unreachable branches by checking the consistency of affine condition sets (the check is known as Omega-test \cite{DBLP:conf/sc/Pught91}).

The operation $\mbox{\opname{Fold}} (\varvar v,e_1,e_2,T)$, where $\varvar v$ is a variable and $e_1$ and $e_2$ are affine expressions, produces S-tree $T'$ that does not contain $\varvar v$ and represents the following function. Given values of all other parameters, $T'$ evaluates to
\begin{equation}
\max \{v \in [e_1,e_2] \mid T(v) \mbox{ evaluates to a term } \}
			\label{eq:fold}
\end{equation}
Making this $T'$ usually involves solving some 1-D parametric integer programming problems and combining the results.

\subsection{Graph Node Structure}\label{ss:eff:node}

In parallel with building the effect of each statement we also compose a graph skeleton, which is a set of nodes with placeholders for future links. For each assignment a separate node is created. At this stage the graph nodes are associated with AST nodes, or statements, in which they were created, for the purpose that will be explained in Section \ref{s:state}. The syntax of a graph node description is presented in Fig.\ref{fig:syngnode}. 

\begin{figure}[htb]
\begin{center}\small
\begin{tabular}{l}
\nt{node} ::=    (\gtag{node} (\nt{name} \nt{context}) \\
 \qquad \qquad \qquad (\gtag{dom} \nt{conditions} )    \\
 \qquad \qquad \qquad (\gtag{ports} \nt{ports})        \\
 \qquad \qquad \qquad (\gtag{body} \nt{computations})  \\
 \qquad \qquad \qquad )                              \\
\nt{context} ::= \nt{names}                          \\
\nt{port} ::= (\nt{name} \nt{type} \nt{source})      \\
\nt{computation} ::= (\gtag{eval} \nt{name} \nt{type} \nt{expression} \nt{destination}) \\
\nt{source} ::= \nt{S-tree} $\mid$ IN                     \\
\nt{destination} ::= \nt{M-tree} $\mid$ OUT               \\
\end{tabular}
\end{center}
\caption {Syntax for graph node description}\label{fig:syngnode} 
\end{figure}

Non-terminals ending with {\em -s} usually denote a repetition of its base word non-terminal, e.g., \nt{ports} signifies \nt{list of ports}. A node consists of a header with name and context, domain description, list of ports that describe inputs and a body that describes output result. The context here is just a list of loop variables surrounding the current AST node. The domain specifies a condition on these variables for which the graph node instance exists. Besides context variables it may depend on structure parameters. Ports and body describe inputs and outputs. The \nt{source} in a port initially is usually an atom $\term{\arname A}{\ser ek}$ (or, generally, an S-tree) depicting an array access ${\arname A}(\ser ek)$, which must be eventually resolved into a S-tree referencing other graph nodes as the true sources of the value (see Section \ref{ss:st:read}).  A computation consists of a local name and type of an output value, an expression to be evaluated, and a destination placeholder $\bot$ which must be replaced eventually by a M-tree that specifies output links (see Section \ref{s:bdf}). The tag IN or OUT declares the node as input or output respectively. 

Consider the statement {\code S=S+X(i)} from program in Fig.8a. The initial view of its graph node is shown in Fig.\ref{fig:inignode}. The expression in \eval\ clause is built from the rhs by replacing all occurrences of scalar or array element with their local names (that became port names as well). A graph node for assignment has a single \eval\ clause and acts as a generator of values written by the assignment. Thus, a term of an effect tree may be considered as a reference to a graph node.
\begin{figure}[htb]
\begin{center}
\begin{tabular}{l}
  (\gtag{node} \garg{(\id{S1}\; i)} \\
  \qquad \qquad   (\gtag{dom} \garg{(i \ge 1) (i\le n)}) \\
  \qquad \qquad   (\gtag{ports} \garg{(\id{s1} \real \term{\id{S}}{})\, (\id{x1} \real \term{\id{X}}{i})})\\
  \qquad \qquad   (\gtag{body} (\gtag{eval} \garg {S \real (\id{s1} + \id{x1})\; \bot}) )\\
  \qquad \qquad   )
\end{tabular}
\end{center}
\caption {An initial view of graph node for statement \code{S=S+X(i)}} \label{fig:inignode}
\end{figure}
\subsection{Processing Non-affine Conditionals}\label{ss:eff:na}

When the source program contains a non-affine conditional statement $X$, special processing is needed. We add a new kind of condition, a predicate function call, or simply predicate, depicted as
\begin{equation}
			\pred{\nt{name}}{\nt{bool-const}}{\nt{L-exprs}}
			\label{eq:pred}
\end{equation}
that may be used everywhere in the graph where a normal affine condition can. It contains a name, sign T or F (affirmation or negation) and a list of affine arguments.

However, not all operations can deal with such conditions in argument trees. In particular, the \opname{Fold} cannot. Thus, we eliminate all predicates immediately after they appear in the effect tree of a non-affine conditional statement.

First, we drag the predicate $p$, which is initially on the top of the effect tree $\Eff A{X} = \cond p{T_1}{T_2}$, downward to leaves. The rather straightforward process is accomplished with pruning. In the result tree, $T_X$, all copies of predicate $p$ occur only in downmost positions of the form $\cond{p}{A_1}{A_2}$, where each $A_i$ is either term or $\bot$. We call such conditional sub-trees \emph{atomic}. In the worst case the result tree will have a number of atomic sub-trees being a multiplied number of atoms in sub-trees $T_1$ and $T_2$.

Second, each atomic sub-tree can now be regarded as an indivisible composite value source. When one of $A_i$ is $\bot$, this symbol depicts an implicit rewrite of an old value into the target array cell ${\arname A}(\ser qk)$ rather than just ``no write''. With this idea in mind we now replace each atomic sub-tree $U$ with a new term $\term{\varvar U_\mathrm{new}}{\ser in}$ where argument list is just a list of variables occurring in the sub-tree $U$. Simultaneously, we add the definition of ${\varvar U_\mathrm{new}}$ in the form of a graph node (associated with the conditional statement $X$ as a whole) which is shown in Fig.\ref{fig:inibnode}. This kind of nodes will be referred to as blenders as they blend two input sources into a single one.
\begin{figure}[htb]
\begin{center}
\begin{tabular}{l}
(\gtag{node} (${\varvar U_\mathrm{new}}$ $i_1\dots i_n$) \\
  \qquad \qquad                (\gtag{dom} $\Dom(X) + \mbox{path-to-$U$-in-$T_X$}$) \\
  \qquad \qquad                (\gtag{ports} ($\id{a}\;t\;\cond{p}{\opname{RW}(A_1)}{\opname{RW}(A_2)}$)\\
  \qquad \qquad                (\gtag{body} (\gtag{eval} $a\:t\:a\:\bot$) )\\
  \qquad \qquad               )
\end{tabular}
\end{center}
\caption {Initial contents of the blender node for atomic subtree $U$ in $\Eff{A}{X}=T_X$} \label{fig:inibnode}
\end{figure}
The domain of the new node is that of statement $X$ restricted by conditions on the path to the sub-tree $U$ in the whole effect tree $T_X$. The result is defined as just copying the input value $a$ (of type $t$). The most intriguing is the source tree of the sole port $a$. It is obtained from the atomic sub-tree $U=\cond{p}{A_1}{A_2}$. Each $A_i$ is replaced (by operator \opname{RW}) as follows. When $A_i$ is a term it remains unchanged. Otherwise, when $A_i$ is $\bot$, it is replaced with explicit reference to the array element being rewritten, $\term{\arname A}{\ser qk}$. However, an issue arises: variables $\ser qk$ are undefined in this context. The only variables allowed here are $\ser in$ (and fixed structure parameters). Thus we need to express indexes $\ser qk$ through ``known'' values $\ser in$.

To resolve this issue consider the list of (affine) conditions $L$ on the path to the subtree $U$ in the whole effect tree $T_X$ as a set of equations connecting variables $\ser qk$ and $\ser in$.  

\begin{proposition}
Conditions $L$ specify a unique solution for values $\ser qk$ depending on $\ser in$. \label{prop2}
\end{proposition}
\begin{proof} 
Consider the other branch $A_j$ of subtree $U$, which must be a term. We prove a stronger statement, namely, that given exact values of all free variables occurring in $A_j$, $\opname{Vars}(A_j)$, all $q$-s are uniquely defined. The term $A_j$ denotes the source for array element ${\arname A}(\ser qk)$ within some branch of the conditional statement $X$. Note, however, that this concrete source is a write on a single array element only. Hence, array element indexes $\ser qk$ are defined uniquely by $\opname{Vars}(A_j)$. Now recall that all these variables are present in the list $\ser in$ (by definition of this list). \qed
\end{proof}

Now that the unique solution does exist, it can be easily found by our affine machinery. See Section \ref{s:bdf} in which the machinery used for graph inversion is described.

Thus, we obtain, for conditional statement $X$, the effect tree that does not contain predicate conditions. All predicates got hidden within new graph nodes. Hence we can continue the process of building effects using the same operations on trees as we did in the purely affine case. Also, for each predicate condition a node must be created that evaluates the predicate value. 
We shall return back to processing non-affine conditionals in Section \ref{ss:bdf:na}.

\section{Evaluation and Usage of States}\label{s:state}
\subsection{Computing States}\label{ss:st:comp}

A state before statement $(s,\varvar I_s)$ in affine program fragment $P$ with respect to array element ${\arname A}(\ser qk)$ is a function that takes as arguments the iteration vector $I_s = (\ser in)$, array indexes $(\ser qk)$ and values of structure parameters and yields the write $(w, I_w)$ in the computation of $P$ that is the last among those that write to array element ${\arname A}(\ser qk)$ before $(s,I_s)$.

In other words this function presents an effect (over array {\arname A}) of executing the program from the beginning up to the point just before $(s,I_s)$. It can be represented as an S-tree, $\State{A}{s}$, called a \emph{state tree} at program point before statement $s$ for array {\arname A}. It can be computed with the following method.



For the starting point of program $P$ we set {\small
\begin{equation}
\State{A}{P}=\cond{\varvar q_1\!\ge \! l_1} 
                  {\cond {\varvar q_1\!\le \! u_1} {\dots \term{\id{A_{ini}}}{q_1,\dots} \dots } {\!\bot} } 
                  {\!\bot}
\label{eq:inistate}
\end{equation}
}where term $\term{\id {A_{ini}}}{\ser qk}$ signifies an untouched value of array element ${\arname A}(\ser qk)$ and $l_i$, $u_i$ are lower and upper bounds of the $i$-th array dimensions (which must be affine functions of fixed parameters). Thus, (\ref{eq:inistate}) means that all {\arname A}'s elements are untouched before the whole program $P$.

The further computation of $\Statef A$ is described by the following production rules:
\begin{enumerate} 
\item 
Let $\State A{B_1;B_2}=T$. 
Then $\State A{B_1}=T$. 
\emph{The state before any prefix of $B$ is the same as that before $B$.}
\item 
Let $\State A{B_1;B_2}=T$. 
Then $\State A{B_2}=\opname{Seq}(T, \Eff A{B_1})$.
\emph{The state after the statement $B_1$ is that before $B_1$ combined by \opname{Seq} with the effect of $B_1$.} 
\item 
Let $\State A{\mbox{\code if $c$ then $B_1$ else $B_2$ endif}} = T$. 
Then $\State A{B_1}=\State A{B_2}=T$. 
\emph{The state before any branch of {\code if}-statement is the same as before the whole {\code if}-statement.}
\item 
Let $\State A{\mbox{\code do $\varvar v=e_1,e_2$; $B$; enddo}} = T$. 
Then $\State A{B}= \opname{Seq}(T, \opname{Fold}(\varvar v,e_1,\varvar v\!-\!1,\Eff A{B}))\label{}$.
\emph{The state before the loop body $B$ with the current value of loop variable $\varvar v$ is that before the loop combined by \opname{Seq} with the effect of all preceding iterations of $B$.}
\end{enumerate}
The form in rule 4 worth some comments. Here the upper limit in the \opname{Fold} clause depends on $\varvar v$. To be formally correct, we must replace all other occurrences of $\varvar v$ in the clause with a fresh variable, say $\varvar v'$. Thus, the resulting tree will (generally) contain $\varvar v$, as it expresses the effect of all iterations of the loop before the $v$-th iteration. 

Using the rules 1-4 one can compute the state in any internal point of the program $P$. The following proposition limits the usage, within a state tree $T$, of terms whose associated statement is enclosed in a conditional with non-affine condition. It will be used further in Section \ref{ss:bdf:na}. 

\begin{proposition}
Let a conditional statement $X$ with non-affine condition be at a loop depth $m$ within a dynamic control program $P$. Consider a state tree $T_p = \State A{p}$ in a point $p$ within $P$ over an array {\arname A}. Let $\term A {\ser {i}{k}}$ be a term in $T_p$, whose associated statement, also $A$, is inside a branch of $X$. Then the following claims are all true:
\begin{itemize}
\item	$m \le k$,
\item	$p$ is inside the same branch of $X$ and
\item	indexes $\ser{i}{m}$ are just variables of loops enclosing $X$.
\end{itemize}
\label{prop:state}
\end{proposition}
\begin{proof}
Let $A$ be a term name, whose associated statement $A$ is inside a branch $b$ of a conditional statement X with non-affine condition. It is either assignment to an array, say {\arname A}, or a blender node emerged from some inner conditional (performing a "conditional assignment" to {\arname A}). From our way of hiding predicate conditions described in Section \ref{ss:eff:na} it follows that the effect tree of $X$, $\Eff A{X}$, as well as of any other statement containing $X$, will not contain a term with name $A$. Hence, due to our way of building states from effects described above, this is also true for the state tree of any point outside $X$, including the state $T_X$ before the $X$ itself. Now, consider the state $T_p$ of a point $p$ within a branch $b_1$ of $X$. (Below we'll see that $b_1 = b$). We have
\begin{equation}
T_p = \opname{Seq} (T_X, T_{X-p}),			 \label{eq:seqstate}
\end{equation}
where $T_{X-p}$ is the effect of executing the code from the beginning of the branch $b_1$ to $p$ (recall that the state before the branch $b_1$, $T_{b_1}$, is the same as $T_X$ according to Rule 3 above). Consider a term $\term A {\ser ik}$ in $T_p$. As it is not from $T_X$, it must be in $T_{X-p}$. 
Obviously, $T_{X-p}$ contains only terms associated with statements of the same branch with $p$. Thus, $b_1=b$. And these terms are only such that their initial $m$ indexes are just variables of $m$ loops surrounding $X$. Thus, given that the operation \opname{Seq} does not change term indexes, we have the conclusion of Proposition \ref{prop:state}. \qed
\end{proof}

\subsection{Resolving Array Accesses}\label{ss:st:read}
Using states before each statement we can build the source graph $\FP$. Consider a graph node $X$. So far, a \nt{source} in its \ports\ clause contains terms representing array access, say $\term{\arname A}{\ser ek}$. Recall that the node $X$ is associated with a  point $p$ in the AST. We take the state $\State A{p}$ and use it to find the source for our access indexes $(\ser ek)$ (just doing substitution followed by pruning). The  resulting S-tree replaces original access term. Doing so with each array access term we obtain the source graph $\FP$.

Recall that each graph node $X$ has a domain $\Dom(X)$ which is a set of possible context vectors. It is specified by a list of conditions, which are collected from surrounding loop bounds and {\code if} conditions. We write $D\implying p$ to indicate that the condition $p$ is valid in $D$ (or follows from $D$). In case of a predicate condition $p=\pred{\varvar p}{b}{\ser ek}$ it implies that the list $D$ just contains $p$ (up to equality of $e_i$). For a S-graph built so far, the following proposition limits the usage of atoms $\term A{\dots}$ whose $\Dom(A)$ has a predicate condition.

\begin{proposition}
Suppose that $B$ is a regular node (not a blender) whose source tree $T$ contains a term $\term A{\ser ik}$ (refering to an assignment to an array {\arname A}). Let {\em $\Dom(A) \implying p$}, where $p=\pred{\varvar p}{b}{\ser jm}$ is a predicate condition. Then:
\begin{itemize}
\item $m\le k$,
\item $j_1=i_1,\dots, j_m=i_m$, and all these are just variables of loops enclosing the conditional with predicate $p$,
\item {\em $\Dom(B)\implying p$}. 
\end{itemize}
\label{prop:src}
\end{proposition}
\begin{proof}
As $\Dom(A)\implying \pred{\varvar p}{b}{\ser jm}$, the predicate $p$ denotes the condition of a conditional statement $X$ enclosed by $m$ loops with variables $\ser jm$, and this $X$ contains the statement $A$ in the branch $b$ (by construction of $\Dom$). The source tree $T$ was obtained by a substitution into the state tree before $B$, $T_B=\State A{B}$, which must contain a term $\term A{\ser {i'}k}$. It follows, by Proposition \ref{prop:state}, that statement $B$ is inside the same branch $b$ (hence, $\Dom(B)\implying p$),  $m\le k$ and $\ser {i'}m$ are just variables $\ser jm$. However the substitution replaces only formal array indexes and does not touches enclosing loop variables, here $\ser jm$. Hence $i'_1=i_1,\dots,i'_m=i_m$. \qed
\end{proof}

When $B$ is a blender the assertion of the Proposition \ref{prop:src} is also valid but $\Dom(B)$ should be extended with conditions on the path from the root of the source tree to the term $\term A{\dots}$. The details are left to the reader. 

\section{Building the Dataflow Model}\label{s:bdf}
\subsection{Inverting S-graph: Affine Case}\label{ss:bdf:u}

In our dataflow computation model a node producing a data element must know 
exactly which other nodes (and by which port) need this data element, and send it to all such ports. 
This information should be placed, in the form of destination M-tree, into the \eval\ clause of each graph node instead of initial placeholder $\bot$. These M-trees are obtained by inversion of S-graph.


First, we split each S-tree into paths. Each path starts with header term $\term R{\ser in}$, containing the list of independent variables, ends with term $\term W{\ser em}$ and has a list of affine conditions interleaved with division clauses like ($e=:k\varvar q+\varvar r$) ($k$ is a literal integer here). In all expressions, variables may be either from header or defined by division earlier. 
The \opname{InversePath} operation produces the inverted path that starts with header term $\term W{\ser jm}$ with new independent variables $\ser jm$, ends with term $\term R{\ser fn}$ and has a list of affine conditions and divisions in between. The inversion is performed by variable elimination. When a variable cannot be eliminated it is simply introduced with clause $(\verb"@" \varvar v)$.

All produced paths are grouped by new headers, each group being an M-tree for respective graph node, in the form $(\&\; T_1\; T_2 \dots)$ where each $T_i$ is a 1-path tree. Further, the M-tree is simplified by the operation \opname{SimplifyTree}. This operation also involves finding bounds for \verb"@"-variables, which are then included into \verb"@"-vertices in the form:
$$(\verb"@"\varvar v (l_1 u_1) (l_2 u_2) \dots T)$$
where $l_i,u_i$ are affine lower and upper bounds of $i$-th interval, and $v$ must belong to one of the intervals. 

\subsection{Inverting S-graph for Programs with Non-affine Conditionals}\label{ss:bdf:na}

When program $P$ has non-affine conditionals the above inversion process will probably yield some M-trees with predicate conditions. The node with such M-tree need an additional port for the value of predicate. We call such nodes \emph{filters}. The simplest filter has just two ports, one for the main value and one for the value of the predicate, and sends the main value to the destination when the predicate value is true (or false) and does nothing otherwise. Splitting nodes with complex M-trees we can always reduce our graph to that with only simplest filters.

Generally, the domain of each arrow and each node may have several functional predicates in the condition list. Normally, an arrow has the same list of predicates as its source and target nodes. However, sometimes these lists may differ by one item. Namely, a filter node emits arrows with a longer predicate list whereas the blender node makes the predicate list one item shorter compared to that of incoming arrow. In the examples below both green (dotted) and red (dashed) arrows have additional predicate in their domain.

However, our aim is to produce not only U-graph, but both S-graph and U-graph which must be both complete and mutually inverse. Thus, we prefer to update the S-graph before inversion such that inversion would not produce predicates in M-trees. To this end, we check for each port whether its source node has enough predicates in its domain condition list. When we see that the source node has less predicates, then we insert a filter node before that port. And the opposite case, that the source has more predicates, is impossible, as it follows immediately from Proposition \ref{prop:src}. 

\section{Examples}\label{s:example}
A set of simple examples of a source program (subroutine) with the two resulting graphs -- S-graph and U-graph -- are shown in Figs. 8,9,11. All graphs were generated as text and then redrawn graphically by hand. Nodes are boxes or other shapes and data dependences are arrows between them. Usually a node has several input and one output ports.
The domain is usually shown once for a group of nodes (in the upper side in curly braces). The groups are separated by vertical line. Each node should be considered as a collection of instance nodes of the same type that differ in domain (context) parameters from each other. Arrows between nodes may fork depending on some condition (usually it is affine condition of domain parameters), which is then written near the start of the arrow immediately after the fork. When arrow enters a node it carries a new context (if it has changed) written there in curly braces. The simplest and purely affine example in Fig.\ref{fig:sum} explains the notations. Arrows in the S-graph are directed from a node port to its source. The S-graph arrows can be interpreted as the flow of requests for input values. (See Section \ref{ss:prop:sre} for details).

\begin{figure}[hb]
\begin{center}
\includegraphics[width=0.47\textwidth]{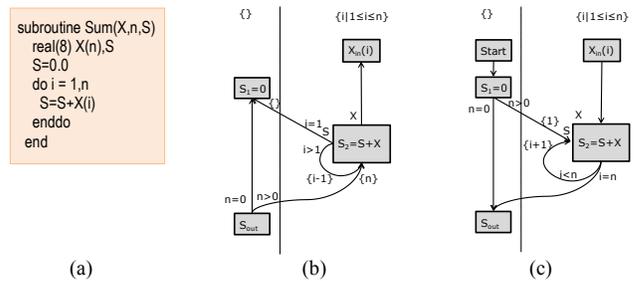}
\end{center}
\caption{Fortran program Sum (a), its S-graph (b) and U-graph (c)}
\label{fig:sum}
\end{figure}

In the U-graph arrows go from node output to node input. In contrast with the S-graph, they denote actual flow of data. The U-graph semantics is described in Section \ref{ss:prop:dflow}.

In the U-graph we need to get rid of zero-port nodes which arise from assignments with constant rhs. We insert into them a dummy port that receives a dummy value. Thus a node {\code Start} sending a token to node S1 appeared in Fig.\ref{fig:sum}c.

A simplest example with non-affine conditions is shown on Fig.\ref{fig:max}. Here appears a new kind of node, the blender, depicted as a blue truncated triangle (see Fig. \ref{fig:max}b). Formally, it has a single port, which receives data from two different sources depending on the value of the predicate. Thus, it has an implicit port for Boolean value (on top). The main port arrows go out from sides; true and false arrows are dotted green and dashed red respectively. 
\begin{figure}
\begin{center}
\includegraphics[width=0.47\textwidth]{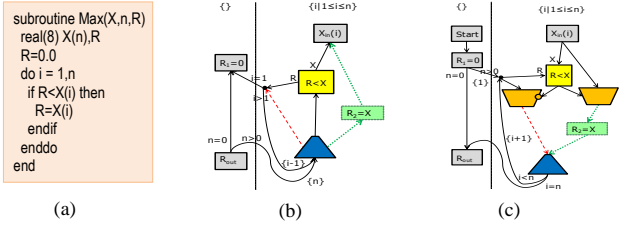}
\end{center}
\caption{Fortran program Max (a), its S-graph (b) and U-graph (c)}
\label{fig:max}
\end{figure}

In the U-graph the blender does not use a condition: in either case it gets a value on its single port without knowing which node has sent it and under which condition. However, as the source itself is not under the needed condition, a filter node must be inserted in between the source node and the receiver port (it is shown in Fig.\ref{fig:sum}c as an inverted orange trapezoid). A circle at the entry point means that the filter is open when the condition is \emph{false}.

A more interesting example, a bubble sort program and its graphs, is shown in Fig.\ref{fig:Bubble}. In contrast with previous ones, this U-graph exhibits high parallelism: the parallel time is $2n$ instead of $n(n+1)/2$ for sequential execution. 

\begin{figure}
\begin{center}
\includegraphics[width=0.47\textwidth]{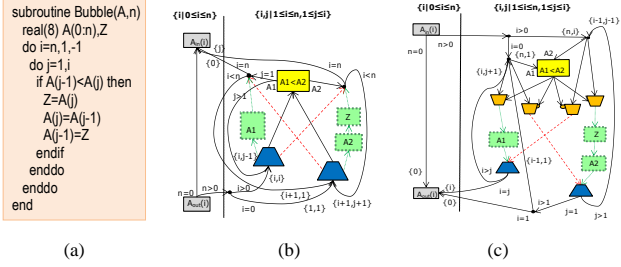}
\end{center}
\caption{Fortran program Bubble (a), its S-graph (b) and U-graph (c)}
\label{fig:Bubble}
\end{figure}

\section{Usage of Polyhedral Model}\label{s:prop}
\subsection{General Form of Dataflow Graph}\label{ss:prop:form}
The general syntax of PM format is shown in Fig.\ref{fig:syngnode}. 
The S-graph is comprised of port source S-trees (single-valued), whereas the U-graph of destination M-trees (multi-valued).
Both graphs must be mutually inverse, i.e., they represent the same dependence relation.

Some nodes produce Boolean values, which can be used as \emph{predicates}. 
In S-graph, they are alowed in a \emph{blender}, which is an identity node with unique port source tree 
of the form $\cond{\pred{\varvar p}{b}{\ser ek}}{T_1}{T_2}$.
In U-graph, we forbid predicates in destination trees, but we allow \emph{filter} nodes, which are in some sense inverse to blenders. Instead of destination tree of the form $\cond {\pred{\varvar p}{b}{\ser ek}}{T_{out}}{\bot}$ they have an additional boolean port $p$ with source $\term P{\ser ek}$ and the destination tree with just $p$ as condition. Thus, filter is a gate which is open or closed depending on the value at port $p$. Note that filters are needed in U-graph, but not in S-graph.

The S-graph must satisfy the two following constraints. The first is a consistency restriction. Consider a node $\term{\id{X}}{\varvar I}$ with domain $D_{\id X}$ and a source tree $T$. Let $I\in D_{\id X}$. Then $T(I)$ is some atom $\term{\id{Y}}{J}$ such that $J\in D_{\id Y}$. The second constraint requires that the S-graph must be \emph{well-founded}, which means that no one object node $\term{\id{X}}{I}$ may transitively depend on itself. 

\subsection{Using the S-graph as a Program}\label{ss:prop:sre}
The S-graph can be used to evaluate output values given all input values (and structure parameters). For simplicity, we assume that each node produces a single output value.

Following \cite{Feautrier:2001:ADA:380466.380472} we transform the S-graph into a system of recurrence equations (SRE), which can be treated as a recursive functional program. 
In Fig.\ref{fig:sre} is presented a SRE for the S-graph from Fig.\ref{fig:max}b.
\begin{figure}
\begin{center}
\newcommand*{\recureq}[2]{$\mathrm{#1=#2}$}
\newcommand*{\If}{if\;}
\newcommand*{\Then}{\;then\;}
\newcommand*{\Else}{\;else\;}
\begin{tabular}{l}
\recureq{P(i)}{R(i) < X(i)} \\
\recureq{B(i)}{\If P(i) \Then X(i) \Else R(i)} \\
\recureq{R(i)}{\If i=1 \Then R1() \Else \If i>1 \Then B(i-1) \Else \bot} \\
\recureq{R1()}{0} \\
\recureq{R_{out}}{\If n=0 \Then R1() \Else \If N>0 \Then B(n) \Else \bot} \\
\end{tabular}
\end{center}
\caption {System of recurrence equations equivalent to S-graph on Fig.\ref{fig:max}b} \label{fig:sre}
\end{figure}
Execution starts with invocation of the output node function. Evaluation step is to evaluate the right hand side calling other invocations recursively.  For efficiency it is worth doing tabulation so that neither function call is executed twice for the same argument list.

Note, that both the consistency and the well-foundedness conditions together provide the termination of the S-graph.

\subsection{Computing the U-graph in the Dataflow Computation Model}\label{ss:prop:dflow}

The U-graph can be executed as program in the dataflow computation model. A node instance with concrete context values \emph{fires} when all its ports get data element in the form of data token. Each fired instance is executed by computing all its \eval\ clauses sequentially. All port and context values are used as data parameters in the execution. In each \eval\ clause the expression is evaluated, the obtained value is assigned to a local variable and then sent out according to the destination M-tree. The tree is executed in an obvious way. In the conditional vertex, the left or right subtree is executed depending on the Boolean value of the condition. In \&-vertices, all sub-trees are executed one after another. An \verb"@"-vertex acts as a {\code do}-loop with specified bounds. Each term of the form $\term{R.x}{\ser fn}$ acts as a token send statement, that sends the computed value to the graph node $R$ to port $x$ with the context built of values of $f_i$. The process stops when all output nodes get the token or when all activity stops (quiescence condition). To initiate the process, tokens to all necessary input nodes should be sent from outside.

\subsection{Extracting Source Function from S-graph}\label{ss:prop:src}
There are two ways to extract the source function from the S-graph.
First, we may use the S-graph itself as a program that computes the source for a given read when the iteration vector of the read as well as values of all predicates are available. 
We take the SRE and start evaluating the term $R(\ser in)$, where $\ser in$ are known integers. We stop as soon as a term of the form $W(\ser jm)$ is encountered on the top level (not inside predicate evaluation), where $W$ is a node name corresponding to a true write operation (not a blender) and $\ser jm$ are some integers.

Also, there is a possibility to extract the general definition of the source function for a given read in a program. 
We start from the term $\term R{\ser{\varvar i}n}$ where $\ser{\varvar i}n$ are symbolic variables and proceed unfolding the S-graph symbolically into just the S-tree. 
Having encountered the predicate node we insert the branching with symbolic predicate condition. 
Having encountered a term $\term W{\ser em}$ for regular assignment statement $W$ we stop unfolding the branch. Having encountered a term for a blender node we unfold it further - this way we avoid taking our artificial dummy assignments as a source.
Proceeding this way we will generate a possibly infinite S-tree (with predicate vertices) representing the source function in question. If we're lucky the S-tree will be finite. 
It seems like in \cite{Feautrier97FADA,GrieblHabThesis} the exact result (in the same sense) is produced only when the above process yields a finite S-tree.

But we get a good result even when the generated S-tree is infinite (this is the case in examples {\code Max} and {\code Bubble}). 
Using a technique like \emph{supercompilation} \cite{DBLP:conf/pdo/Turchin85} it is possible to fold the infinite S-tree into a finite cyclic graph.

\section{Related Work}\label{s:relwrk}

The foundations of dataflow analysis for arrays have been well established in the 90-s by Feautrier \cite{Feautrier88parametricinteger,Feautrier91dataflowanalysis}, Pugh\cite{DBLP:conf/sc/Pught91},  Maslov\cite{DBLP:conf/popl/Maslov94} and others. 
Their methods use the Omega and PIP libraries and yield an exact dependence relation for any pair of read and write references in affine program. 
Thus, our work adds almost nothing for the affine case (besides producing a program in the dataflow computation model). 
However, for non-affine conditions, the state-of-the-art is generally a \emph{fuzzy} solution \cite{Feautrier97FADA}, in which the source function produces \emph{a set} of possible sources. 
The authors claim that nothing more can be done. But all depends on the form we want to see the result in. Sometimes one may be satisfied with the source function expressed in the form of a {finite quast extended with predicate vertices}. Then why not allow a bit more general form - a S-graph with predicate nodes, or SRE? The main thing is that it was good for something.

Known translations from WDP to KPN \cite{StefanovThesis} usually rely on FADA and seem to succeed only when FADA succeeds to be exact (judging by the examples used). It is interesting what and how they do with the Bubble Sort in Fig.\ref{fig:Bubble}.

Our base affine machinery of building the exact PM also differs. While it is common to consider each read-write pair separately and then combine the results, our method first produces effects and states using only writes, and then resolves each read against the respective state. 
It is interesting to compare our effect/state building process with that of backward traversing the control flow graph \cite{GrieblHabThesis}. 
Both processes move along the same path but in opposite directions. 
Authors usually argue for moving backward noting that the process can stop when the total source is found (cf. also \cite{DBLP:conf/popl/Maslov94}). 
We hope to obtain the same effect just implementing our algorithm in a lazy language. Then the tree $T$ will not be built at all in calls like $\opname{Seq}(T,t)$, where $t$ is a term.

In principle, the way we deal with non-affine conditionals can be reformulated as follows: (1) push all dynamic {\code if}s to the innermost level; (2) add {\code else} parts which just assign the existing value to the same variable; (3) do FADA, identifying different copies of each predicate value; (4) collect the resulting exact source functions as S-graph, or SRE; (5) use this SRE as recursive definition of true source (passing by all dummy assignments introduced at item 2). 
I am grateful to the IMPACT reviewers for pointing out this relationship.

\section{Conclusion}\label{s:concl}
Our aim was to convert a program $P$ of a specific class into the dataflow computation model. 
Thus we need not only to build the exact and complete polyhedral model (which is a set of exact source functions for all reads in $P$), but also to invert it and thus obtain the exact use set function for each write in $P$. 
The latter form can be used as a program in a dataflow computation model. 
A prototype translator is implemented in Refal-6 \cite{Refal:6}.
It admits an arbitrary WDP.


The work was supported by Russian Academy of Sciences Presidium Program for Fundamental Research "Fundamental Problems of System Programming" in 2009--2014.

%
%
\bibliography{PM,Paral,PEMC-en}
\bibliographystyle{plain}
\clearpage

\clearpage

\end{document}